\documentclass[a4paper,reqno]{amsart}

\usepackage[utf8]{inputenc}
\usepackage{amsmath,amssymb,amsfonts,amsthm}
\usepackage{tikz}
\usetikzlibrary{calc}
\usepackage[normalem]{ulem}
\usepackage{hyperref}

\usepackage{thmtools}
\usepackage{thm-restate}

\usepackage[normalem]{ulem}
\usepackage{hyperref}

\theoremstyle{plain}
\newtheorem{theorem}{Theorem}
\newtheorem{lemma}[theorem]{Lemma}
\newtheorem{corollary}[theorem]{Corollary}
\newtheorem{question}[theorem]{Question}

\theoremstyle{definition}
\newtheorem{definition}[theorem]{Definition}
\newtheorem{remark}[theorem]{Remark}
\newtheorem{principle}[theorem]{Principle}
\newtheorem{proposition}[theorem]{Proposition}

\newcommand{\UU}{\mathcal{U}}
\newcommand{\UUp}{\UU_\bullet}

\newcommand{\proj}{\mathsf{p}_1}

\newcommand{\Prop}{\mathsf{hProp}}
\newcommand{\refl}{\mathsf{refl}}

\newcommand{\ct}{%
  \mathchoice{\mathbin{\raisebox{0.5ex}{$\displaystyle\centerdot$}}}%
             {\mathbin{\raisebox{0.5ex}{$\centerdot$}}}%
             {\mathbin{\raisebox{0.25ex}{$\scriptstyle\,\centerdot\,$}}}%
             {\mathbin{\raisebox{0.1ex}{$\scriptscriptstyle\,\centerdot\,$}}}
}

\newcommand{\istype}[1]{\mathsf{is}\mbox{-}{#1}\mbox{-}\mathsf{type}}
\newcommand{\isprop}{\mathsf{is}\mbox{-}\mathsf{prop}}

\newcommand{\Susp}{\mathbf \Sigma}

\newcommand{\trunc}[2]{\mathopen{}\left\Vert #2\right\Vert_{#1}\mathclose{}}
\newcommand{\brck}[1]{\trunc{}{#1}}
\newcommand{\tproj}[3][]{\mathopen{}\left|#3\right|_{#2}^{#1}\mathclose{}}
\newcommand{\bproj}[1]{\tproj{}{#1}}
\newcommand{\defeq}{:\equiv} 
\newcommand{\isequiv}{\mathsf{isequiv}}

\newcommand{\bool}{\mathbf 2}
\newcommand{\nil}{\mathsf{nil}}

\newcommand{\base}{\mathsf{base}}

\newcommand{\unit}{\mathsf{unit}}
\newcommand{\cons}{\mathsf{cons}}
\newcommand{\icons}{\mathsf{icons}}
\newcommand{\iseqconstructor}{\mathsf{iseq}}

\newcommand{\emptyt}{\ensuremath{\mathbf{0}}}
\newcommand{\unitt}{\ensuremath{\mathbf{1}}}

\newcommand{\loops}{\mathsf{loops}}

\newcommand{\concat}{+\hspace*{-7pt} +}

\newcommand{\signed}[1]{{#1}^{\pm}}

\usepackage{accents}
\newcommand*{\dt}[1]{%
  \accentset{\mbox{\large\bfseries .}}{#1}}

\newcommand{\invert}[1]{\dt {#1}}

\newcommand{\Red}{\mathsf{Red}}
\newcommand{\auxRed}{\mathsf{R}}
\newcommand{\red}{\mathsf{red}}
\newcommand{\List}{\mathsf{List}}
\newcommand{\LS}[1]{\List(\signed #1)}
\newcommand{\N}{\mathbb N}
\renewcommand{\concat}{}
\newcommand{\length}{\mathsf{length}}

\newcommand{\inl}{\mathsf{inl}}
\newcommand{\inr}{\mathsf{inr}}

\newcommand{\sm}[1]{\Sigma \left(#1\right).\,}

\newcommand{\swap}{\mathsf{sw}}
\newcommand{\overlap}{\mathsf{ov}}
\newcommand{\trconst}{\mathsf{tr}}

\newcommand{\redIsNeutral}{\mathsf{red}\mbox{-}\mathsf{is}\mbox{-}\mathsf{neutral}}

\newcommand{\Sn}[1]{\mathbb{S}^{#1}}

\newcommand{\FG}[1]{\mathbf{F}(#1)}
\newcommand{\FGwoArg}{\mathbf{F}}
\newcommand{\NR}[1]{\mathbf{N}(#1)}

\newcommand{\north}{\mathsf{N}}
\newcommand{\south}{\mathsf{S}}
\newcommand{\merid}{\mathsf{mer}}

\newcommand{\Wedge}{\mathsf{W}}

\newcommand{\LNF}[1]{\mathsf{LNF}(#1)}

\begin{document}

\title{Free Higher Groups in Homotopy Type Theory}

\author{Nicolai Kraus \and Thorsten Altenkirch}
\date{}

\thanks{This work was supported by 
the Engineering and Physical Sciences Research Council (EPSRC),
grant reference EP/M016994/1,
and by USAF, Airforce office for scientific research, award
FA9550-16-1-0029.\\
\emph{Disclaimer as requested by publisher:} This paper has been published in the proceedings of LiCS'18, \url{https://doi.org/10.1145/3209108.3209183}.}

\begin{abstract}
 Given a type $A$ in homotopy type theory (HoTT), we can define the \emph{free $\infty$-group on $A$} as
 the loop space of the suspension of $A + 1$.
 Equivalently, this free higher group can be defined as a higher inductive type $\FG{A}$ with constructors $\unit : \FG{A}$,
 $\cons : A \to \FG{A} \to \FG{A}$, and conditions saying that every $\cons(a)$ is an auto-equivalence on $\FG{A}$.
 Assuming that $A$ is a set (i.e.\ satisfies the principle of unique identity proofs), we are interested in the question whether $\FG{A}$ is a set as well,
 which is very much related to an open problem in the HoTT book~\cite[Ex.~8.2]{hott-book}.
 We show an approximation to the question, namely that the fundamental groups of $\FG{A}$ are trivial, i.e.\ that $\trunc 1 {\FG{A}}$ is a set.
\end{abstract}

\maketitle

\section{Introduction}

An important feature of \emph{Martin-L\"of type theory} (MLTT) is the identity type which makes it possible to express equality inside type theory.
More precisely, if $A$ is a type (in any context), and $x,y : A$ are terms, then $\mathsf{Id}_A(x,y)$ is a type whose inhabitants represent proofs that $x$ and $y$ are equal.
As it is common nowadays, we write $x=_A y$ or $x=y$ instead of $\mathsf{Id}_A(x,y)$, and call its elements simply \emph{equalities}.
\emph{Homotopy type theory} (HoTT) embraces the fact that $x=y$ may come with interesting structure.
This means that, in many cases, we do not only care about the question whether $x$ and $y$ are equal, but also \emph{how} or \emph{in which ways} they are equal.
For example, due to Voevodsky's univalence axiom, $\bool =_\UU \bool$ is equivalent to $\bool$.
Here, $\bool$ is the type of booleans, $\UU$ is a type universe, and the two equalities stem from the two ways in which $\bool$ is equivalent to itself.
Another non-trivial example is $\base =_{\Sn 1} \base$ which turns out to be equivalent to the type of integers~\cite{Licata:2013:CFG:2591370.2591407}, where $\Sn 1$ is the ``circle'' in HoTT.

A type where each such type of equalities can have at most one inhabitant is called a \emph{set}, and further said to satisfy the principle of \emph{unique identity proofs} (UIP).
An example for a set is the type of natural numbers: $0 =_\N 1$ is equivalent to the empty type $\emptyt$, while $0 =_\N 0$ is equivalent to the unit type $\unitt$, and so on.
Many algebraic structures can be implemented straightforwardly if we are happy to do everything with sets; for example, \cite[Def~6.11.1]{hott-book} defines 
a set-level group to be a tuple $(G, \circ, u, \cdot^{-1}, \ldots)$, where $G$ is a set, together with a multiplication operation $\circ : G \times G \to G$, a unit element $u$, an inversion operator $G \to G$, and equalities expressing the usual laws. 
The book then further shows how one can construct the free set-level group over a given set (see \cite[Chp 6.11]{hott-book}).

More interesting and challenging is it to define higher-level structures, not restricted to sets.
Since we use HoTT (rather than, say, set theory) as our foundation, it is natural to attempt this.
For example, it has been known for some time that externally, every type carries the structure of an \emph{$\infty$-groupoid}~\cite{lumsdaine_weakOmegaCatsFromITT,bg:type-wkom}.
Internally, we can play a trick and use the following definition, which probably can be considered ``HoTT folklore'':

\begin{definition} \label{def:cheating-group-def}
 An \emph{$\infty$-group} is a type $G$ which is equivalent to a loop space.
 More precisely, $G$ is a group means that we have a connected type $X$ and a point $x:X$ together with an equivalence $G \simeq (x=x)$.
 If $G$ is an $\infty$-group represented by $(X,x)$ and $H$ is a second $\infty$-group represented by $(Y,y)$, then a \emph{homomorphism} $G \to_{\infty\mathsf{grp}} H$ is a pointed function $(X,x) \to_\bullet (Y,y)$.
\end{definition}

In other words, an \emph{$\infty$-group} is a type that admits a \emph{delooping}.
Clearly, the unit element of this group is $\refl_x$, and composition is given by composition of equalities.
Some theory of higher groups in homotopy type theory has very recently been developed by Buchholtz, van Doorn and Rijke~\cite{buchVDrijke:highergroups}.

It is worth noting that Definition~\ref{def:cheating-group-def} makes use of the fact that a suitable notion of an (untruncated) group already naturally exists in HoTT, which is not the case for many other interesting structures.
Defining untruncated algebraic structures in general and \emph{directly} is an important open problem in HoTT.
To see why this is hard, let us start from the set-level definition $(G,\circ, u, \cdot^{-1}, \ldots)$ above and remove the condition that $G$ is a set.
The equalities which guarantee that the multiplication is associative, the unit is neutral, and inverses cancel, are not sufficient anymore; they would not give a well-behaved definition of a higher group. 
For example, one may ask oneself how one would prove that $x \circ (y \circ (z \circ w))$ equals $((x \circ y) \circ z) \circ w$: there are two canonical ways, and  these should coincide (``MacLane's pentagon''), but any such rule that we add would have to satisfy its own coherences.
It is currently unknown whether it is possible to complete this sort of definition in a satisfactory way.
In a nutshell, the problem is that the usual definitions would, if expressed internally in type theory, amount to infinite structures of coherences. In classical homotopy theory, these conditions are often organised in the form of an \emph{operad} \cite{sta:h-spaces,ada:infinite-loop-spaces}, but a representation of such structures that can be written down in type theory has not been discovered so far.
This is certainly not for a lack of trying; cf.\ the much-discussed open problem of defining \emph{semisimplicial types}~\cite{uf:semisimplicialtypes}.

In this paper, we study the \emph{free $\infty$-group} over a type $A$.
It is folklore in homotopy type theory that a suitable definition of the free $\infty$-group over $A$ is given by the loop space of a
``wedge of $A$-many circles'', or put differently, the suspension $\Susp (A + \unitt)$.
For us, it will however be more helpful to give a more explicit construction of this free higher group which we call $\FG A$.
We define $\FG A$ to be the \emph{higher inductive type} (HIT)~\cite[Chp 6]{hott-book} which as constructors has a neutral element $\unit : \FG A$ and a multiplication operation $\cons : A \to \FG A \to \FG A$, together with conditions ensuring that each $\cons(a) : \FG A \to \FG A$ is an equivalence.
This definition encodes the a priori infinite tower of coherence condition suitably and it will turn out that it is equivalent to the loop space of $\Susp(A + \unitt)$.

The most basic properties that one would expect from a free higher group are easy to prove.
More intriguing is the question what the free $\infty$-group has to do with the free set-based group.
Clearly, we would want the former to be a generalisation of the latter.
The most obvious way of interpreting this is to ask whether, for a \emph{set} $A$, the free higher group $\FG A$ coincides with the construction of the set-based higher group.
It turns out that the central question is the following:

\begin{question} \label{eq:main-question}
 If $A$ is a set, is $\FG{A}$ a set as well?
\end{question}
One reason why we believe that Question~\ref{eq:main-question} is hard is that a slight generalisation of it is a known open problem in HoTT which has been recorded in the book (see~\cite[Ex.~8.2]{hott-book}).
To be precise, the open problem asks whether, for a set $A$, the suspension $\Susp (A)$ is a $1$-type; our question is equivalent to asking whether $\Susp (A + \unitt)$ is a $1$-type.
A positive answer to the open problem would imply a positive answer to our Question~\ref{eq:main-question}, but we do not expect that our question is fundamentally easier.
(Recall from the book \cite{hott-book} that the suspension $\Susp(A)$ is the HIT with constructors $\north, \south : \Susp(A)$ and $\merid : A \to \north = \south$, for ``north'', ``south'', ``meridian''.)

The core of our paper consists of a proof of a weakened version, a first approximation, of Question~\ref{eq:main-question}.
Our main result can be phrased as follows:

\begin{restatable}{theorem}{mainthm}
\label{thm:mainthm}
 If $A$ is a set, then all fundamental groups of $\FG{A}$ are trivial.
 In other words, $\trunc 1 {\FG{A}}$ is a set.
\end{restatable}

Our strategy to prove this is to define a simple reduction system together with a \emph{non-recursive approximation}, written $\NR A$, to the free $\infty$-group.
These are both based on the usual construction of the free monoid on $A$, that is, $\List(A)$.
The proof of Theorem~\ref{thm:mainthm}, with all the tools and strategies that we need to develop, constitutes the main part of the paper.

The reason why our strategies are not sufficient to provide a full answer to Question~\ref{eq:main-question} is that $\NR A$ is really only an approximation.
Defining $\FG A$ in a non-recursive way (i.e.\ without using some sort of induction that quantifies over elements of $\FG A$ itself)
seems to, as least as far as we can see, correspond exactly to expressing the infinite coherence tower ``directly''.
We would not be surprised if it turned out that Question~\ref{eq:main-question} was in fact independent of ``standard HoTT'' (the type theory developed in the book~\cite{hott-book}), and if the status of Question~\ref{eq:main-question} was related to the status of semisimplicial types.
We will come back to this in the conclusions of the paper.

\paragraph*{Setting}
The type theory that we consider in this paper is the standard homotopy type theory developed in the book~\cite{hott-book}.
This means that we have \emph{univalent universes} $\UU$, \emph{function extensionality}, and \emph{higher inductive types} (HITs).
Regarding notation, we strive to stay close to \cite{hott-book}, with the exception that we write $(a : A) \to B(a)$ instead of $\Pi(a:A).B(a)$.
All performed constructions will preserve the universe level, hence there is no risk at omitting it.
Uncurrying is done implicitly, allowing us to write $f(a,b)$ instead of $f(a)(b)$.

\paragraph*{Outline}
We give the precise definition of the free $\infty$-group $\FG A$ in Section~\ref{sec:freeHigherGroups}, together with
some simple observations.
The statements in this section (apart from the definition of $\FG A$ and its universal property) are not important for the main part of the paper, the proof of Theorem~\ref{thm:mainthm}, which is given in Section~\ref{sec:mainsection}.
As a corollary of the constructions, we get that $\trunc 1 {\FG A}$ coincides with the set-based construction of the free group, and $\FG A$ does under the assumption that Question~\ref{eq:main-question} has a positive answer.
In Section~\ref{sec:conclusions}, we make some concluding remarks and discuss related open problems in homotopy type theory.

\section{Free $\infty$-Groups} \label{sec:freeHigherGroups}

\subsection{Definition and First Properties}

Let us start with an explicit construction of the free higher group as a higher inductive type $\FG A$, since this is the central concept of the paper.
We use a point constructor $\unit : \FG A$ for the neutral element,
and a constructor $\cons : A \to \FG A \to \FG A$ which ``multiplies'' an element of $A$ with any other group element.
We write $\cons_a : \FG A \to \FG A$ instead of $\cons(a)$ since, most of the time, we regard $a$ as a fixed variable.
The trick which completes the definition in an elegant way, due to Paolo Capriotti, is
to add the condition that $\cons_a$ is an equivalence for every $a$.
This cannot be done directly (at least not according to the usual intuitive rules for presentations of higher inductive types),
but it can be ``unfolded'' and expressed via a suitable collection of constructors.
Let us show the concrete definition.
\begin{definition} \label{def:FA}
 The \emph{free $\infty$-group} over a given type $A$ is the following higher inductive type $\FG{A}$:
 \begin{align*}
  & \textbf{data} \; \FG{A} \; \textbf{where} \\
  & \; \unit : \FG{A} \\
  & \; \cons : A \to \FG{A} \to \FG{A} \\
  & \; \icons : A \to \FG{A} \to \FG{A} \\
  & \; \mu_1 : (a : A) \to (x : \FG{A}) \to \icons_a (\cons_a (x)) = x \\
  & \; \mu_2 : (a : A) \to (x : \FG{A}) \to \cons_a (\icons_a (x)) = x \\
  & \; \mu : (a : A) \to (x : \FG{A}) \to \mathsf{ap}_{\cons(a)}(\mu_1(a,x)) = \mu_2(a,\cons_a(x))
 \end{align*}
\end{definition}
At first glance, the above HIT appears complicated and rather unappealing.
Due to the ``unfolding'', the underlying idea that we have discussed above is somewhat hidden.
The constructors $\unit$ and $\cons$ are the standard ones that one would use to define the type of lists over $A$, $\List(A)$, or, in other words, the free monoid over $A$.
The remaining four constructors $(\icons, \mu_1, \mu_2, \mu)$ simply say that, for every $a : A$, the function $\cons_a : \FG{A} \to \FG{A}$ is a \emph{half-adjoint equivalence} as defined in \cite[Chp 4]{hott-book}.
This is the ``unfolding'' mentioned before; note that we could equally well have used other definition of equivalences, such as the ``bi-invertible'' or ``contractible fibre'' constructions.
In any case, this means that we can think of $\FG{A}$ as being fully described as a triple $(\unit, \cons, \iseqconstructor)$, with $\iseqconstructor : (a : A) \to \isequiv(\cons_a)$.

To make use of the type $\FG A$, we need to know an elimination principle for it.
This can be stated as an induction (dependent elimination) principle, which is how it is done in the book~\cite{hott-book}.
More concise, and (we would say) conceptually clearer, is the approach of phrasing it using a universal property, in other words, a recursion (non-dependent) elimination principle with a uniqueness property.
The equivalence between these approaches for inductive types has been discussed by Awodey, Gambino, and Sojakova \cite{awodeyGamSoja_indTypesInHTT}, for some HITs, by Sojakova \cite{DBLP:journals/corr/Sojakova14}, and a restricted version for set-truncated HITs can be found in~\cite{ACDKNF}.
For \emph{concrete} HITs, such as our $\FG A$, it is straightforward to derive the various elimination principles from each other.
We state the universal property using the presentation as a \emph{homotopy-initial algebras}~\cite{awodeyGamSoja_indTypesInHTT}:

\begin{principle} \label{principle:hinitial}
 We say that an \emph{$\FG{A}$-algebra structure} on a type $X$ consists of a point $u : X$, a map $f : A \to X \to X$, and a proof $p : (a: A) \to \isequiv(f(a))$.
 We say that the type of \emph{$\FG{A}$-algebra morphisms} between $(X,u,f,p)$ and $(Y,v,g,q)$ consists of triples $(h,r,s)$, where $h : X \to Y$, $r : f(u) = v$, and $s : h \circ f = g \circ h$.
 Then, the induction principle of $\FG{A}$ is equivalent to saying that $(\FG{A}, \unit, \cons, \iseqconstructor)$ is \emph{homotopy initial}, i.e.\ that for any $(X,u,f,p)$, the type of morphisms from $(\FG{A}, \unit, \cons, \iseqconstructor)$ to $(X,u, f,p)$ is contractible.
\end{principle}
We will come back to $\FG A$-algebras later.

An obvious question is whether $\FG A$ really deserves to be called the \emph{free $\infty$-group} on $A$.
There are two points: first, we need to check that it is a higher group in the sense of Definition~\ref{def:cheating-group-def}, and second, we have to justify the attribute \emph{free}.

For the first point, note that the suspension $\Susp(A + \unitt)$ has an equivalent description which can be obtained by essentially collapsing the point $\south$ with the path given by the unit type:
it is the HIT $\Wedge(A)$ with a single point constructor $\north : \Wedge(A)$ and a family of loops indexed over $A$, as in $\loops : A \to \north = \north$, a \emph{wedge of $A$-many circles}.
Note that $\Wedge(A)$ is automatically connected.
A further side remark is that $\Wedge(\unitt)$ is canonically equivalent to the circle $\Sn 1$.
We then observe:

\begin{lemma} \label{lem:FA-is-group}
 The free $\infty$-group $\FG{A}$ is an $\infty$-group, with $\Wedge(A)$ as its delooping.
 The canonical equivalence $e : \FG A \to (\north =_{\Wedge(A)} \north)$ maps the structure as one would expect, i.e.\ 
 we have $e(\unit) = \refl$ and $e(\cons_a(x)) = \loops(a) \ct e(x)$.
\end{lemma}
 Note that this statement is completely independent from the rest of the paper.
\begin{proof}
 This is a relatively straightforward generalisation of the proof that the loop space of $\Sn 1$ is equivalent to the type of integers.
 The proof is an application of \cite[Lem 8.9.1]{hott-book} and does not provide much insight, which is why we choose to omit it.
 For a detailed argument, one can easily adapt the proof given by Brunerie (for an only slightly different statement) in~\cite[Sec~6]{DBLP:journals/corr/abs-1710-10307}.
\end{proof}
From the above lemma, we can in particular observe that $\FG \unitt$ is a presentation of the type of integers.

Next, we need to justify why we call $\FG A$ the \emph{free} higher group.
The following presentation of the argument was suggested by Paolo Capriotti.
Let us consider the following diagram:
\begin{equation} \label{eq:adj-diag}
 \begin{tikzpicture}[x=1cm,y=1cm,baseline=(current bounding box.center)]
  \node (A) at (0,0) {$\UU$};
  \node (B) at (2,0) {$\UUp$};
  \node (C) at (4,0) {$\UUp$};
  \node (AB) at (1,0) {$\bot$};
  \node (BC) at (3,0) {$\bot$};
  \draw[->, bend left=30] (A) to node [above] {$+ \unitt$} (B);
  \draw[->, bend left=30] (B) to node [above] {$\Susp$} (C);
  \draw[->, bend left=30] (B) to node [below] {$\proj$} (A);
  \draw[->, bend left=30] (C) to node [below] {$\Omega$} (B);
 \end{tikzpicture}
\end{equation}
Here, $\UUp$ is the universe of pointed types.
The function $(+ \unitt) : \UU \to \UUp$ maps a type $X$ to $(X+\unitt, \inr(\star))$, while projection $\proj$ simply forgets the point.
As in \cite[Chp 6.5]{hott-book}, we regard the suspension as a function $\Susp : \UUp \to \UUp$, mapping $(X,x)$ to $(\Susp(X),\north)$, and $\Omega$ is the loop space.
For $X : \UU$ and $Y,Z : \UUp$, it is easy to see that there is a canonical equivalence
\begin{equation} \label{eq:fst-adj}
 \left(X \to \proj Y\right) \simeq \left((X + \unitt) \to_\bullet Y\right),
\end{equation}
and by \cite[Lem 6.5.4]{hott-book}, we have
\begin{equation} \label{eq:snd-adj}
 \left(Y \to_\bullet \Omega(Z)\right) \simeq \left(\Susp(Y) \to_\bullet Z\right).
\end{equation}
The above diagram \eqref{eq:adj-diag} should for our purpose only be regarded as an illustration of these two equivalences.
Talking about the adjunctions more precisely is difficult since the correct notions would be $\infty$-categorical.
This leads into a territory that is vastly unexplored in homotopy type theory~\cite{capKra_semisegal}, although higher adjunctions can be represented using only a finite amount of data~\cite{Riehl2016802}; here, we do not go further into this.

Let $G$ be a given $\infty$-group, represented by $(Z,z)$.
This means that we have $G \simeq (z=z)$.
We can then calculate:
\begin{equation}
 \begin{alignedat}{3}
  &&&& \quad & \FG A \to_{\infty\mathsf{grp}} G \\
  && \mbox{by Def \ref{def:cheating-group-def} and Lem~\ref{lem:FA-is-group}} \quad &\simeq && (\Wedge(A), \north) \to_\bullet (Z,z) \\
  && \mbox{} &\simeq && \Susp(A + \unitt) \to_\bullet (Z,z) \\
  && \mbox{by \eqref{eq:snd-adj}}  \quad &\simeq && (A + \unitt) \to_\bullet \Omega(Z,z) \\
  && \mbox{by \eqref{eq:fst-adj}}  \quad &\simeq && A \to (z=z) \\
  && \mbox{} &\simeq && A \to G.
 \end{alignedat}
\end{equation}
Thus, $\FGwoArg : \UU \to \infty\mathsf{GRP}$ is ``morally'' left adjoint to the forgetful functor which returns the underlying type of a higher group.

\subsection{On Alternative Constructions} \label{subsec:alternatives}

As a preparation for the development in Section~\ref{sec:mainsection}, and to better understand the difficulties with $\FG A$, let us attempt to construct $\FG A$ in a different way.
Let us write $\signed A$ for $A + A$; we call $\signed A$ the type of \emph{elements of $A$ with a sign}, and we think of $\inl(a)$ as $a$ and $\inr(a)$ as $a^{-1}$.
For $a : \signed A$, we write $\invert a$ for the element we get by changing the sign.
Of course, this means that $\invert{\invert a} = a$.

Elements of the free group $\FG{A}$ are, at least intuitively, lists over $\signed A$.
The difficulty is that different lists may represent the same group element.
This happens, for example, for $[\inl(a), \inr(a)]$ (i.e.\ $a \cdot a^{-1}$) and the empty list, both of which represent the unit of the group.
We can avoid this problem by quotienting to identify the list $[x_0, \ldots, x_k, a, \invert a, x_{k+1}, \ldots, x_n]$ with the list $[x_0, \ldots, x_n]$.
This quotient will be a set by definition of the quotient (set quotient) operation.
If we are happy to work only with sets, and to set-truncate everything, then this is entirely possible,
and in fact, it is a construction of the set-based free group given in \cite[Thm 6.11.7]{hott-book}.
If, like in this paper, we do not want to restrict ourselves to sets, we might think of taking a HIT which has path constructors for each such pair of lists, without set-truncating.
The problem is that we need coherences:
if we use a path constructor to reduce one redex and then a second, we should get the same equality as if we reduce the second redex and then the first.
When looking at \emph{three} redexes, we need to express that these equalities ``fit together'', and so on.
This is an instance of the problem of infinite coherences which seem to be hard and possibly impossible to express in HoTT.
In Section~\ref{sec:mainsection}, we will perform a finite approximation of this construction in order to show Theorem \ref{thm:mainthm}, although we will see that a couple of additional arguments are required to complete the proof.

Alternatively, we could think to only consider lists over $\signed A$ in \emph{normal form}, i.e.\ lists which come together with a proof that they do not contain a redex.
The type of lists over $A$ in normal form is a set (assuming that $A$ is a set), and the presentation is indeed fully coherent.
The trouble is that we are in general unable to define a suitable binary operation on this set, i.e.\ we are lacking a group operation.
If we have two lists in normal form, their concatenation might not be in normal form, and for arbitrary types, we have no way of calculating a normal form or even checking whether we already have a normal form.

Unsurprisingly, the approach with normal forms works if $A$ has decidable equality:
\begin{proposition}
 If $A$ has decidable equality, in the sense that
 \begin{equation}
  (a_1, a_2 : A) \to (a_1 = a_2) + ((a_1 = a_2) \to \emptyt),
 \end{equation}
 then $\FG A$ has decidable equality as well.
 Moreover, $\FG A$ is in this case canonically equivalent to the set-truncated construction of the free group as given in \cite[Chp 6.11]{hott-book}.
\end{proposition}
\begin{proof}[Proof sketch]
 Thanks to \cite[Thm 6.11.7]{hott-book}, we can take set-quotiented lists (as described above) as the definition of the set-truncated free group.
 Using decidable equality of $A$, it is easy to see that this quotient is equivalent to the type of lists in normal form; let us write $\LNF A$ for the latter type.
 An element of $\LNF A$ is a list together with a propositional property, and we have an embedding $\LNF A \to \LS A$.
 What is left to do is to compare $\LNF A$ with $\FG A$.
 Note that $\LNF A$ is a set without being explicitly set-truncated.
 There is a canonical $\FG A$-algebra structure on $\LNF A$, giving rise to a map $\FG A \to \LNF A$.
 Further, one can construct a function $\LS A \to \FG{A}$, by induction on the list.
 The empty list is mapped to $\unit$, $\inl(a)$ translates into an application of $\cons_a$, and $\inr(a)$ becomes $\icons_a$.
 These functions give rise to an equivalence $\LNF A \simeq \FG A$, and since $\LNF A$ has decidable equality, $\FG A$ enjoys the same property.
\end{proof}

\section{The fundamental group of the free $\infty$-group} \label{sec:mainsection}

In this section, the core of the paper, we develop a couple of techniques that, when combined, allow us to prove Theorem~\ref{thm:mainthm}.
For the whole section, let us assume that $A$ is a given set.
Given lists $x,y : \LS A$, we write $xy$ for their \emph{concatenation}, i.e.\ the list we get by simply joining the two lists as in $[a_1,a_2] [a_3, a_4] = [a_1,a_2,a_3,a_4]$.
Since this operation is associative (up to a canonical and fully coherent equality), we omit brackets and write $xyz$ for both $(xy)z$ and $x(yz)$.
Given $a : \signed A$, we regard $a$ as a one-element list and allow ourselves to write e.g.\ $xayz$ or $a\invert a y$.

\subsection{A simple reduction system in type theory}

As discussed in Section~\ref{subsec:alternatives} above,
we can think of elements of $\FG A$ as lists over $\signed A$,
and the main problem is that different lists represent the same group element.
This motivates the development of a system of reductions.
\begin{definition} \label{def:Red}
The type family
 $\Red : \LS A \to \UU$,
which expresses that a list represents the same group element as the empty list (i.e.\ the neutral element of the group $\FG{A}$), is defined as follows.
We first define an auxiliary family $\auxRed : \N \to \LS A \to \UU$ by induction on the natural numbers:
\begin{alignat*}{3}
  & \auxRed_0 (x) &&\defeq \; \; && \length(x) = 0 \\
  & \auxRed_{2+n}(x) && \defeq && \Sigma (a : \signed A), (y,z : \LS A), \\
  &                         &&&& \phantom{\Sigma} \left(\length(y) + \length(z) = n\right), \\
  &                         &&&& \phantom{\Sigma} \left(x = y \concat a \concat \invert a \concat z\right), \\
  &                         &&&& \phantom{\Sigma} \auxRed_n (y \concat z)
\end{alignat*}
 Using this, we set $\Red(x) \defeq \auxRed_{\length(x)}(x)$.
\end{definition}

If we have indexed inductive families in the theory, we can alternatively define $\Red$ directly as such a family generated by
\begin{alignat*}{2}
 & \mathsf{zero} : &\;& \Red(\nil) \\
 & \mathsf{step} : &\;& (y,z : \LS A) \to (a : \signed A) \to \Red(y \concat z) \to \Red(y \concat a \invert a \concat z).
\end{alignat*}
The two definitions are essentially the same, only represented in different ways.
In both cases, given $r : \Red(x)$, 
we say that $r$ witnesses that $x$ can be \emph{reduced} to the empty list and
we call $r$ a \emph{reduction sequence}.
We view it as a sequence consisting of \emph{steps}, each of which removes a single \emph{redex} $a \invert a$.
An example of a reduction sequence $r: \Red(a \invert a b c \invert c \invert b)$ could be pictured as follows, where each step is represented by an arrow $\leadsto$ annotated with the redex it reduces:
\begin{equation} \label{eq:sample-red-sequence}
 a \invert a b c \invert c \invert b \; \stackrel{c \invert c}\leadsto \; a \invert a b \invert b \; \stackrel{a \invert a}\leadsto \; b \invert b \; \stackrel{b \invert b}\leadsto \; \nil.
\end{equation}

\begin{remark}
 There are a couple of points that we want to point out explicitly.
 \begin{enumerate}
  \item In the above example and in the discussions to come, $a : \signed A$ is already positive or negative, which means that every redex is of the form $a \invert a$; the possibility $\invert a a$ is already covered.
  \item The number of steps of $r : \Red(x)$ is simply half of the length of the list $x$, which means that all elements of $\Red(x)$ have the same number of steps.
  In particular, it is easy to prove that $\Red(x)$ is empty if $\length(x)$ is odd.
  \item For a given list $x$, there is no way to \emph{compute} a reduction sequence, since we do not know whether an occurring pair $bc$ forms a redex. 
  A reduction $r: \Red(x)$ encodes equalities which guarantee that all redexes that it reduces are really redexes.
  Deciding whether $bc$ is a redex would require decidable equality on $A$ (but of course, we can always check whether an element of $\signed A$ is positive or negative, and this analysis might give us that $bc$ is definitely not a redex).
  \item For a given $x$, equality on $\Red(x)$ \emph{is} decidable.
  This is because a sequence encodes the positions of the redexes that it reduces, and positions are decidable, while the (in general undecidable) equalities on $\signed A$ are propositions.
  Similarly, if we have $r : \Red(xby)$, we can say in which step $b$ is reduced, since this is encoded by a position.
 \end{enumerate}
\end{remark}

Let us remind ourselves that the goal of the paper is to show that $\FG A$ has trivial fundamental groups.
This is a statement about equalities between equalities.
If we think of a reduction sequence as a proof that a list represents the neutral group element,
i.e.\ as something giving rise to an equality proof, it is hopefully intuitive that we now want to discuss the relationship between different reduction sequences.
In a nutshell, we want to give a criterion which guarantees that two reduction sequences give rise to equal equalities.
To do so, we consider \emph{transformations}:

\begin{definition} \label{def:transformations}
 Let $w : \LS A$ be a list and $r : \Red(w)$ a reduction sequence.
 We consider the following two operations, each of which allows us to create a new reduction sequence in $\Red(w)$ from $r$:
 \begin{enumerate}
  \item Swap two consecutive independent steps in $r$. More precisely, if $r$ is a sequence of the form
   \begin{equation} \label{eq:seq-1-to-transform}
    \ldots \; \leadsto xa \invert a y b \invert b z \; \stackrel{b \invert b}\leadsto \; xa \invert a yz \; \stackrel{a \invert a}\leadsto \; xyz \leadsto \ldots,
   \end{equation}
   we can change it to
   \begin{equation} \label{eq:seq-2-to-transform}
    \ldots \; \leadsto xa \invert a y b \invert b z \; \stackrel{a \invert a}\leadsto \; xyb \invert bz \; \stackrel{b \invert b}\leadsto \; xyz \leadsto \ldots.
   \end{equation}
   Analogously, we can change \eqref{eq:seq-2-to-transform} into \eqref{eq:seq-1-to-transform}.
  \item If a step reduces a redex $a \invert a$ in a list of the form $xa \invert a ay$, we can change this step to remove the redex $\invert a a$ instead, or vice versa.
  This means that
   \begin{equation}
    \ldots \; \leadsto xa \invert a ay \; \stackrel{a \invert a}\leadsto \; xay \; \leadsto \ldots
   \end{equation}
   can be changed to
   \begin{equation}
    \ldots \; \leadsto xa \invert a ay \; \stackrel{\invert aa}\leadsto \; xay \; \leadsto \ldots,
   \end{equation}
   or vice versa.
 \end{enumerate}
 We say that $r : \Red(w)$ can be \emph{transformed} into $s : \Red(w)$ if there is a finite chain of these operations that changes $r$ into $s$.
\end{definition}

After what we said in the paragraph before Definition~\ref{def:transformations}, the best we could hope for is that \emph{any reduction sequence can be transformed into any other reduction sequence} (of the same list $w$).
Indeed, this is what we will show.
We start with a technical lemma which will not only help us to prove what we just said (Corollary~\ref{cor:transform-any}), but also another useful consequence (Corollary~\ref{cor:not-stuck}).

\begin{lemma} \label{lem:normalising-reductions}
 Assume we are given a list of the form $x a \invert a y$, i.e.\ a list in $\signed A$ with an explicitly given redex $a \invert a$.
 Assume further that we have a reduction sequence $s : \Red(x a \invert a y)$.
 It is possible to transform $s$ into a reduction sequence which reduces the redex $a \invert a$ in the first step, i.e.\ starts with $xa \invert ay \; \stackrel{a \invert a}\leadsto \; xy \; \leadsto \; \ldots$.
\end{lemma}
\begin{proof}
 Let $x$, $a$, $y$, and $r : \Red(xa \invert ay)$ be given.
 Let us write $m$ for the number of the step in which $a$ is reduced, and $n$ for the number of the step in which $\invert a$ is reduced.
 There are three cases:
 \begin{itemize}
  \item If $m = n$, then the redex $a \invert a$ is reduced in step $n$.
  If $n = 0$, there is nothing to do.
  Otherwise, we can swap this step with step $(n-1)$, since the two steps will be independent of each other.
  Swapping a further $(n-1)$ times, we can move the step reducing $a \invert a$ to the beginning of the sequence.
  \item If $m > n$, then $a \invert a$ are not reduced together, but $\invert a$ is reduced with some $\invert {\invert a}$ to its right instead. Note that $\invert {\invert a} = a$.
  Before step $n$, the list thus has to be of the form $u a \invert a a v$, and step $n$ consists of reducing $\invert a a$.
  We define $r_1$ to be the reduction sequence which is identical to $r$ in every step expect in step $n$ where it reduces $a \invert a$; this is the second of the two possible operations in Definition~\ref{def:transformations}.
  We are now in case one ($m=n$). 
  \item The case $m < n$ is analogous to the case $m > n$. \qedhere
 \end{itemize}
\end{proof}

\begin{corollary} \label{cor:transform-any}
 Any reduction sequence can be transformed into any other reduction sequence.
 More precisely, for $w : \LS A$ and $r,s : \Red(w)$, we can transform $r$ into $s$.
\end{corollary}
\begin{proof}
 A reduction sequence is given by a chain of reduction steps,
 and the number of steps in $r$ and $s$ are equal (both are $\mathsf{length}(w)/2$).
 Thus, it is sufficient to transform $r$ into a sequence which consists of the same steps as $s$.
 By the above lemma, we can transform $r$ into a sequence $r'$ which in the first step reduces whichever redex $s$ reduces in the first step.
 Applying the same argument to the ``tail'' of the sequences (note that $r'$ and $s$, each with the first step removed, still reduce the same list), we get a transformation into a sequence which in every step mirrors the reduction of $s$ and is thus equal to $s$.
\end{proof}

A second easy consequence is that, if a list is reducible, then we cannot ``get stuck'' while reducing:
we can start reducing at an arbitrary position without risking of ending up with an unreducible list.
Note that we write $B \leftrightarrow C$ for $(B \to C) \times (C \to B)$. 
\begin{corollary} \label{cor:not-stuck}
For any lists $y, z$ and $a : \signed A$, we have
\begin{equation}
 \Red(y \concat a \concat \invert a \concat z) \leftrightarrow \Red(y \concat z).
\end{equation}
\end{corollary}
\begin{proof}
 The direction $\leftarrow$ is immediate, by adding a single reduction step reducing $a \invert a$.
 The direction $\rightarrow$ is an application of Lemma~\ref{lem:normalising-reductions}.
\end{proof}

\begin{remark}
 Note that Corollary~\ref{cor:transform-any} subtly but crucially depends on the assumption that $A$ is a set, while Lemma~\ref{lem:normalising-reductions}, as formulated, would work for arbitrary types $A$.
 It is true independently of $A$ that a reduction sequence is given by a chain of reduction steps.
 A reduction step encodes the \emph{position} at which the reduction is taking place (say, the length of the list $y$ in Definition~\ref{def:Red}), together with a proof that the reduction is possible (i.e.\ a proof that the pair at the position is actually a redex).
 The second part amounts to an equality in $\signed A$ (since ``$ab$ being a redex'' means $a = \invert b$); thus, it is a proposition if $A$ is a set.
 In this case, a reduction step is determined by the position, and a reduction sequence is determined by the chain of positions which it encodes.
 The proof of Corollary~\ref{cor:transform-any} relies on this.
 
 Lemma~\ref{lem:normalising-reductions} holds even without the requirement of $A$ being a set.
 However, note that the proof of Lemma~\ref{lem:normalising-reductions}, when it uses the second operation in Definition~\ref{def:transformations}, has to construct a new equality (this is hidden in the sentence ``Before step $n$, the list thus has to be of the form $u a \invert a a v$'').
 Therefore, the new sequence constructed in Lemma~\ref{lem:normalising-reductions} \emph{will} reduce $a \invert a$ in the first step, but the proof that $a \invert a$ is indeed a redex could be a nontrivial one.
\end{remark}

\subsection{A non-recursive approximation to the free $\infty$-group}

We are ready to define a \emph{non-recursive approximation} to the free group $\FG{A}$, a HIT that we call $\NR{A}$.
By \emph{non-recursive}, we mean that constructors of $\NR{A}$ do not use points or paths of $\NR{A}$ in their arguments.

\begin{definition} \label{def:NA}
We define $\NR{A}$ to be the HIT with the following constructors:
\begin{equation*}
 \begin{alignedat}{2}
  &\eta : &\; & \LS A \to \NR{A} \\
  &\tau : && (x : \LS A) \to (a : \signed A) \to (y : \LS A) \\
          &&& \; \to \eta(x \concat a \concat \invert a \concat y) = \eta(x \concat y) \\
  &\swap : && (x : \LS A) \to (a : \signed A) \to (y : \LS A) \\
          &&& \; \to (b : \signed A) \to (z : \LS A) \\
          &&& \; \to \tau(x,a,yb \invert b z) \ct \tau(xy,b,z) = \tau(xa \invert a y, b, z) \ct \tau(x,a,yz) \\
  &\overlap : && (x : \LS A) \to (a : \signed A) \to (y : \LS A) \\
          &&& \; \to \tau(x,a,ay) = \tau(xa,\invert a,y) \\
  &\trconst : &&\istype{1}(\NR{A}) 
 \end{alignedat}
\end{equation*}
\end{definition}

We can think of $\NR{A}$ (without the last constructor) as a ``wild'' quotient of $\LS A$.
Recall that we said that lists over $\signed A$ correspond to very intentional representations of group elements.
The HIT with constructors $\eta$ and $\tau$ can be thought of as a ``level $0$ approximation'' to a fully coherent non-recursive quotient of $\LS A$:
we identify some lists which represent the same group element, but the equalities are incoherent.
This is partially remedied by the constructors $\swap$ (``swap'') and $\overlap$ (``overlap''), ensuring that the equalities generated by $\tau$ satisfy basic coherence.
They can be pictured as follows:

 \begin{equation}
  \begin{tikzpicture}[x=4cm,y=-1.5cm,baseline=(current bounding box.center)]
   \node (XAYBZ) at (0,0) {$\eta(xa \invert a y b \invert b z)$};
   \node (XYBZ) at (1,0) {$\eta(xy b \invert b z)$};
   \node (XAYZ) at (0,1) {$\eta(xa \invert a yz)$};
   \node (XYZ) at (1,1) {$\eta(xyz)$};
   \node (SW) at (0.5,0.5) {$\swap(x,a,y,b,z)$};
   
   \draw[->] (XAYBZ) to node [above] {$\tau(x,a,yb\invert b z)$} (XYBZ);
   \draw[->] (XAYBZ) to node [left] {$\tau(xa \invert ay, b, z)$} (XAYZ);
   \draw[->] (XYBZ) to node [right] {$\tau(x,a,yb\invert b z)$} (XYZ);
   \draw[->] (XAYZ) to node [below] {$\tau(xa \invert ay, b, z)$} (XYZ);
  \end{tikzpicture}
 \end{equation}

 \begin{equation}
  \begin{tikzpicture}[x=4cm,y=-2cm,baseline=(current bounding box.center)]
   \node (XAAY) at (0,0) {$\eta(xa \invert a a y)$};
   \node (XAY) at (1,0) {$\eta(x a y)$};
   \node (OV) at (0.5,0) {$\overlap(x,a,y)$};
   
   \draw[->, bend left=20] (XAAY) to node [above] {$\tau(x,a,ay)$} (XAY);
   \draw[->, bend right=20] (XAAY) to node [below] {$\tau(xa, \invert a, y)$} (XAY);
  \end{tikzpicture}
 \end{equation}

$\swap$ and $\overlap$ themselves are not directly guaranteed to be coherent; if we omit the constructor $\trconst$, we can think of $\NR{A}$ as a ``level $1$ approximation''.
$\trconst$ ensures that all higher equalities hold, by forcing $\NR{A}$ to be $1$-truncated.
The statement that $\NR{A}$ is an approximation to the free higher group can then be made by drawing a connection to $\trunc 1 {\FG{A}}$,
which we will do later.

If a list can be reduced, then in $\NR{A}$, it is indistinguishable from the empty list:
\begin{lemma} \label{lem:redIsNeutral}
We have a function
\begin{equation}
 \redIsNeutral : (z : \LS A) \to \Red(z) \to \eta(z) = \eta(\nil).
\end{equation}
\end{lemma}
\begin{proof}
 We need to analyse the element $r : \Red(z)$.
 It encodes a finite number of reduction steps.
 The first reduction step shows that $z$ is of the form $z = x a \invert a y$, thus the constructor $\tau(x,a,y)$ (transported along the equality $z=xa \invert a y$) provides us with the equality
 $\eta(z) = \eta(xy)$.
 Similarly, each of the remaining reduction steps encoded in $r$ shows how $\tau$ can be applied, and the concatenation of all these equalities yields $\eta(z) = \eta(\nil)$.

 If $\Red$ is defined as an indexed inductive family, $\redIsNeutral(x)(r)$ can be constructed by induction on $r$, and the induction step is given by the constructor $\tau$.
\end{proof}

Not only can we show that reducible lists are equal to $\nil$ in $\NR{A}$, it is also the case that the concrete witness of reducibility does not matter:
\begin{lemma} \label{lem:redIsNeutral-wconst}
 For any given $x$, the function 
 \begin{equation}
  \redIsNeutral(x) : \Red(x) \to \eta(x) = \eta(\nil) 
 \end{equation}
 is weakly constant, in the following sense:
 \begin{equation}
  (r, s : \Red(x)) \to \redIsNeutral(x)(r) = \redIsNeutral(x)(s).
 \end{equation}
\end{lemma}
\begin{proof}
 The constructors $\swap$ and $\overlap$ ensure that, if two reduction sequences can be transformed into each other, then they lead to equal proofs of $\eta(x) = \eta(\nil)$.
 More precisely,
 the first operation in Definition~\ref{def:transformations} is exactly covered by the constructor $\swap$, while the second operation is covered by $\overlap$.
 The statement thus follows from Corollary~\ref{cor:transform-any}.
\end{proof}

The point of $\NR{A}$ is that it is easier to reason about $\NR{A}$ than about $\FG{A}$, thanks to the absence of recursive constructors; one can say that $\NR{A}$ attempts to bridge the gap between $\LS A$ and $\FG{A}$.
We first define a property stating that an element of $\NR{A}$ can be reduced.
We write $\Prop$ for $\sm{X:\UU} \isprop(X)$ as usual.
\begin{lemma}
 The family $\brck - \circ \Red : \LS A \to \Prop$ extends to a family $\red : \NR{A} \to \Prop$ as in the following commuting triangle:
 \begin{equation}
  \begin{tikzpicture}[x=2.5cm,y=-1.2cm,baseline=(current bounding box.center)]
   \node (LSA) at (0,0) {$\LS A$};
   \node (TYPE) at (1,0) {$\UU$};
   \node (PROP) at (2,0) {$\Prop$};
   \node (NA) at (0,1) {$\NR{A}$};
   
   \draw[->] (LSA) to node [above] {$\Red$} (TYPE);
   \draw[->] (TYPE) to node [above] {$\brck -$} (PROP);
   \draw[->] (LSA) to node [left] {$\eta$} (NA);
   \draw[->, dotted] (NA) to node [below] {$\red$} (PROP);
  \end{tikzpicture}
 \end{equation}
\end{lemma}
\begin{proof}
 We do induction on $\NR{A}$.
 Clearly, we have to set $\red(\eta(x)) \defeq \brck{\Red(x)}$.
 The proof obligation of the constructor $\tau$ is met by Corollary~\ref{cor:not-stuck}.
 The remaining two constructors are trivial, since they ask for equalities between elements of propositions.
\end{proof}

To avoid confusion with elements of $\LS A$, which we call $x,y,z,\ldots$, we use Greek letters for elements of $\NR{A}$.
If $\gamma : \NR{A}$ is reducible, it is equal to the neutral element:
\begin{lemma} \label{lem:NA-red-equal}
There is a function of type
\begin{equation}
 (\gamma : \NR{A}) \to \red(\gamma) \to \gamma = \eta(\nil).
\end{equation}
\end{lemma}
\begin{proof}
 We do induction on $\gamma$.
 First, we consider the case $\gamma \equiv \eta(x)$, and we want to find $f_x : \red(\eta(x)) \to \eta(x) = \eta(\nil)$.
 Recall that a weakly constant function into a set (which the codomain here is) factors through the propositional truncation~\cite{lmcs:3217},
 hence since $\red(\eta(x)) \equiv \brck{\Red(x)}$ by definition,
 Lemma~\ref{lem:redIsNeutral-wconst} gives us a function
 \begin{equation}
  f_x : \red(\eta(x)) \to \eta(x) = \eta(\nil)
 \end{equation}
 such that $f_x(\bproj r) = \redIsNeutral(x)(r)$.
 We want to extend this function to $\NR{A}$.
 Induction on $\gamma$ requires us to provide constructions corresponding to $\tau$, $\swap$, and $\overlap$.
 The latter two are contractible, and we do not need to worry about them.
 The proof obligation for $\tau$ says that, for any $y,a,z$, and witnesses $s : \red(\eta(yz))$, $s' : \red(\eta(ya \invert a z))$, the triangle
 \begin{equation}
  \begin{tikzpicture}[x=4.5cm,y=-1.3cm,baseline=(current bounding box.center)]
   \node (T00) at (0,0) {$\eta(y a \invert a z)$};
   \node (T10) at (1,0) {$\eta(yz)$};
   \node (T11) at (1,1) {$\eta(\nil)$};
   
   \draw[->] (T00) to node [below left] {$f_{y a \invert a z} (s')$} (T11);
   \draw[->] (T10) to node [right] {$f_{yz}(s)$} (T11);
   \draw[->] (T00) to node [above] {$\tau(y,a,z)$} (T10);
  \end{tikzpicture}
 \end{equation}
 commutes.
 This is a proposition, thus we can assume that $s$, $s'$ come from actual reduction sequences, i.e.\ we have $r : \Red(yz)$ with $\bproj r = s$ and $r' : \Red(ya \invert a z)$ with $\bproj{r'} = s'$.
 This simplifies the triangle to:
 \begin{equation}
  \begin{tikzpicture}[x=4.5cm,y=-1.3cm,baseline=(current bounding box.center)]
   \node (T00) at (0,0) {$\eta(y a \invert a z)$};
   \node (T10) at (1,0) {$\eta(yz)$};
   \node (T11) at (1,1) {$\eta(\nil)$};
   
   \draw[->] (T00) to node [below left] {$\redIsNeutral(y a \invert a z)(r')$} (T11);
   \draw[->] (T10) to node [right] {$\redIsNeutral(yz)(r)$} (T11);
   \draw[->] (T00) to node [above] {$\tau(y,a,z)$} (T10);
  \end{tikzpicture}
 \end{equation}
 If we write $r'' : \Red(ya \invert a z)$ for the sequence $r$, extended by the single step reducing $a \invert a$ in the beginning,
 we see that composition of the horizontal and vertical arrow give $\redIsNeutral(ya \invert a z)(r'')$.
 Thus, Lemma~\ref{lem:redIsNeutral-wconst} yields the required commutativity of the triangle.
\end{proof}

This allows us to conclude:
\begin{lemma} \label{lem:NA-locally-set}
 $\NR{A}$ is a set locally at $\eta(\nil)$, in the sense that ${\eta(\nil) =_{\NR{A}} \eta(\nil)}$ is contractible.
\end{lemma}
\begin{proof}
 If equality is implied by a ``reflexive mere relation'', then the type is a set (\cite[Thm~7.2.2.]{hott-book}, sometimes called ``Rijke's theorem'').
 Here, we need the local formulation of this statement as given in \cite{nicolai:thesis},
 together with Lemma~\ref{lem:NA-red-equal}.
\end{proof}

Next, we want to extend this observation and show that $\NR A$ is a set.
In general, if a type $X$ is a set locally at $x_0 : X$ and we have an equivalence $e : X \to X$, then $X$ is also a set locally at $e(x_0)$,
using that $\mathsf{ap}_e$ will be an equivalence.
Therefore, if for every $y: X$ there is an equivalence mapping $x_0$ to $y$ (we can say that such an $X$ is ``homogeneous''),
then $X$ is a set; and in fact, it is enough if for a given $y$ the equivalence merely exists (i.e.\ hidden with a truncation).
This is one motivation for the following technical lemma, where we construct equivalences $\NR A \to \NR A$.
Another motivation is that these equivalences are the main part of an $\FG A$-algebra structure on $\NR A$, but we will come back to this later.

\begin{lemma} \label{lem:NA-has-equiv}
 There is a function $f : \signed A \to \NR A \to \NR A$ such that, for every $c : \signed A$, the map $f_c : \NR A \to \NR A$ is an equivalence with $f_{\invert c}$ as its inverse.
 Further, the construction can be done such that, for every $x : \LS A$, we have $f_c(\eta(x)) \equiv \eta(cx)$.
\end{lemma}
\begin{proof}
 Let $c:A$ be given.
 We need to define $f_c : \NR{A} \to \NR{A}$, i.e.\ for a given $\alpha : \NR{A}$, we need $f_c(\alpha) : \NR{A}$.
 This can be done by recursion on $\alpha$ 
 in the obvious way:
 \begin{itemize}
  \item We set $f_c(\eta(x)) \defeq \eta(cx)$.
  \item Next, we need a witness of $f_c(\eta(xa \invert a y)) = f_c(\eta(xy))$. Slightly abusing notation, we write $f_c(\tau(x,a,y))$ for this\footnote{The more accurate notation might be $\mathsf{ap}_{f_c}(\tau(x,a,y))$.}, and we set $f_c(\tau(x,a,y) \defeq \tau(cx,a,y)$.
  \item Similarly, we set $f_c(\swap(x,a,y,z)) \defeq \swap(cx,a,y,z)$;
  \item and $f_c(\overlap(x,a,y)) \defeq \overlap(cx,a,y)$;
  \item and finally, we have $f_c(\trconst) \defeq \trconst$.
 \end{itemize}
 We need to show that $f_c$ is an inverse of $f_{\invert c}$.
 It is sufficient to show that, for $\alpha : \NR{A}$, we have $f_c(f_{\invert c}(\alpha)) = \alpha$ and $f_{\invert c}(f_c(\alpha)) = \alpha$.
 Let us concentrate on the first of these, as the second is no more than a copy which switches the sign of $c$.
 Note that the goal is an equality in the $1$-type $\NR{A}$ and thus a set.
 Thus, when we do induction on $\alpha$, in order to construct a function $h : (\alpha : \NR{A}) \to f_c(f_{\invert c}(\alpha)) = \alpha$, the proof obligations for $\swap$, $\overlap$, and $\trconst$ are trivial.
 For $\eta$ and $\tau$, the constructions work as follows:
 \begin{itemize}
  \item For $\eta$, we need $h(\eta(x))$ of type $f_c(f_{\invert c}(\eta(x))) = \eta(x)$, which reduces to $\eta(c \invert c x) = \eta(x)$.
  Therefore, we can set $h(\eta(x)) \defeq \tau(\nil, c,x)$.
  \item For $\tau$, we need to construct $h(\tau(x,a,y))$ which shows that $h(\eta(xa \invert a y))$ and $h(\eta(xy))$ are equal as paths over $\tau(x,a,y)$.
  
  After unfolding what this means, we see that the type of $h(\tau(x,a,y))$ is: 
  \begin{equation}
   \mathsf{ap}_{f_{c \invert c}}(\tau(x,a,y)) \ct \tau(\nil,c,xy) = \tau(\nil,c,xa\invert a y) \ct \tau(x,a,y).   
  \end{equation}

  By the given construction of $f_c$ above, this simplifies to
  \begin{equation}
   \tau(c \invert c x,a,y) \ct \tau(\nil,c,xy) = \tau(\nil,c,xa \invert a y) \ct \tau(x,a,y),
  \end{equation}
  which is given by $\swap(\nil,c,x,a,y)$. \qedhere
  \end{itemize}
\end{proof}

From Lemma \ref{lem:NA-has-equiv}, it is very easy to derive an $\FG A$-algebra structure. We will record this later in Corollary \ref{cor:NA-has-FA-structure}.
Before going there, we draw another immediate conclusion:
\begin{lemma} \label{lem:NA-is-set}
 The type $\NR A$ is a set.
\end{lemma}
\begin{proof}
 It suffices to show the that, for any given $\alpha : \NR A$, the type ${\alpha =_{\NR A} \alpha}$ is contractible.
 We do induction on $\alpha$.
 Since the goal is a proposition which becomes trivial for all higher constructors, we only need to show the statement for the point constructor $\eta$.
 Thus, assuming $x : \LS A$, we need to show that $\eta(x) =_{\NR A} \eta(x)$ is contractible.
 We do induction again, this time on the list $x$. If $x$ is the empty list $\nil$, then the statement is given by Lemma~\ref{lem:NA-locally-set}.
 Otherwise, $x$ is $ay$ with $a : \signed A$.
 Consider the equivalence $f(a) : \NR A \to \NR A$ from Lemma~\ref{lem:NA-has-equiv}.
 It gives us an equivalence $\mathsf{ap}_{f(a)} : \eta(y) = \eta(y) \to \eta(ay) = \eta(ay)$, the domain of which is contractible by the induction hypothesis.
\end{proof}

\subsection{Connection between approximations of the free group}

In order to make use of $\NR{A}$ and the results we have found so far, we show in this section that $\trunc 1 {\FG{A}}$ is 
equivalent to $\NR A$.
A direct proof via ``maps in both directions which are inverse to each other'' would in principle be possible.
Our calculations however led to a very messy argument, which did not provide much insight.
In this paper, we therefore proceed a bit differently:
after constructing $\FG A$- and $\NR A$-algebra structures on both $\NR A$ and $\trunc 1 {\FG A}$ (which corresponds to constructing the two functions), we show that the structures are ``compatible'', i.e.\ that a certain $\NR A$-algebra map is also an $\FG A$-algebra morphism.
We will later explain in detail what this means.

Recall from the statement of Principle~\ref{principle:hinitial} that an \emph{$\FG{A}$-algebra structure} on a type $X$ consists of a point $u : X$ and a family $f : A \to X \to X$ such that each $f_a$ is an equivalence on $X$, witnessed by some $p: (a : A) \to \mathsf{isequiv}(f_a)$.
An $\FG A$-algebra is a type $X$ with such a structure, i.e.\ a tuple $(X,u,f,p)$.
Also recall that an $\FG A$-algebra morphism between $(X,u,f,p)$ and $(Y, v, g, q)$ is a triple $(h, r , s)$, where $h : X \to Y$, $r : f (u) = v$, and $s : h \circ f = g \circ h$.
Similarly, we say that a type $Y$ carries an $\NR A$-algebra structure if we have a tuple $(e,t,s,o,h)$ mirroring the constructors of $\NR A$,
with $\NR A$-algebra morphisms defined in the obvious way.
Then, $(\NR A, \eta, \tau, \swap, \overlap, \trconst)$ is homotopy initial among all $\NR A$-algebras.

From Lemma~\ref{lem:NA-has-equiv}, we immediately get a canonical $\FG A$-algebra.
Note that here and later we write $\_$ (blank) for a ``nameless'' component which should be clear from the context.
\begin{corollary} \label{cor:NA-has-FA-structure}
 We have an $\FG A$-algebra $(\NR A, \eta(\nil), \overline{f}, \_)$, where $\overline f$ is given by the function Lemma~\ref{lem:NA-has-equiv} composed with the embedding $\inl : A \to \signed A$ of $a$ into ``positively signed $a$''.
 Since $\FG A$ carries the initial such structure, we get a canonical map $\FG A \to \NR A$.
\end{corollary}

It does not seem to be the case in general that truncations preserve algebra structure, since this seems to require a choice principle; see e.g.\ the infinitary branching trees in~\cite{alt-kap:tt-in-tt,alt-dan-kra:partiality}.
Fortunately, it is very simple in our case:
\begin{lemma}\label{lem:trFA-has-FA-str}
 The type $\trunc{1}{\FG A}$ carries an $\FG A$-algebra structure, and $\bproj - : \FG A \to \trunc{1}{\FG A}$ is an $\FG A$-algebra morphism.
\end{lemma}
\begin{proof}
 This follows easily from the fact that $\bproj -$ preserves equivalences.
\end{proof}

Corollary \ref{cor:NA-has-FA-structure} can be reversed if we add a truncation:
\begin{lemma} \label{lem:FG-has-NR-structure} 
 The type $\trunc{1}{\FG A}$ carries an $\NR A$-algebra structure.
\end{lemma}
\begin{proof}
 Doing this in detail is tedious, but there is no hidden difficulty.
 The components corresponding to the constructors $(\eta, \tau, \swap, \overlap)$ could all be constructed using $\FG A$ directly, we simply need to throw in $\bproj - : \FG A \to \trunc{1}{\FG A}$ at the right places.
 The component corresponding to $\eta$, which has type
 \begin{equation}
  e : \LS A \to \trunc{1}{\FG A}, 
 \end{equation}
 is simply given by composing instances of $\bproj \cons$ or $\bproj \icons$ with each other, one for each element of the list $x$ (we use $\bproj \cons$ for positive list elements and $\bproj \icons$ for negative ones), and applying them on the unit element $\bproj \unit$.
 We write $e(x) \defeq \bproj{\vec {\cons_x}}$ for this.
 For example, if $x$ is the list $a\invert bc$ (where $a$, $b$, $c$ are now all assumed to be positive), 
 then $e(x) \equiv \bproj{\vec{\cons_x}} \equiv \bproj{\cons_a}(\bproj{\icons_b}(\bproj{\cons_c}(\bproj \unit)))$.

 The component corresponding to $\tau$, which has type
 \begin{equation}
  t : (x : \LS A) \to (a : \signed A) \to (y : \LS A) \to e(xa \invert a y) = e(xy),
 \end{equation}
 is then given by ``whiskering'' as in (let us for simplicity assume that $a$ is positive):
 \begin{equation}
  t(x,a,y) \defeq \mathsf{ap}_{\bproj{\vec{\cons(x)}}}(\bproj{\mu_2}(a,y))
 \end{equation}
 The components for $\swap$ and $\overlap$ are essentially naturality of whiskering and $\mu$, respectively, while the fact that we have $1$-truncated $\FG A$ gives us the component for the constructor $\trconst$.
\end{proof}

Using that $\FG A$ carries the (homotopy) initial $\FG A$-algebra structure, and $\NR A$ the (homotopy) initial $\NR A$-algebra structure,
the statements of Corollary~\ref{cor:NA-has-FA-structure} and Lemma \ref{lem:FG-has-NR-structure} give us maps $h$ and $k$ as follows:

\begin{equation} \label{eq:two-maps}
 \begin{tikzpicture}[x= 1.15cm , y = -.5cm,baseline=(current bounding box.center)]
   \node[] (A) at (0,0) {$\FG A$};
   \node[] (B) at (3,0) {$\NR A$};
   \node[] (C) at (6,0) {$\trunc{1}{\FG A}$};
   \draw[->] (A) to node[above] {map of $\FG A$-algs} (B);
   \draw[->] (B) to node[above] {map of $\NR A$-algs} (C);
   \node[] (L1) at (1.5,0.5) {$h$};
   \node[] (L2) at (4.5,0.5) {$k$};
 \end{tikzpicture}
\end{equation}

In the next lemma, we show that \emph{both} these functions are maps of $\FG A$-algebras.
This will be sufficient to show that $\NR A$ is a retract of $\trunc{1}{\FG A}$.
It was a suggestion by Paolo Capriotti that this lemma might lead to a cleaner proof of the property we ultimately want, which, we think, is indeed the case.
\begin{lemma} \label{lem:map-is-double-alg}
 The map $k$ in \eqref{eq:two-maps} is a map of $\FG A$-algebras, with respect to the $\FG A$-algebra structures constructed in Corollary \ref{cor:NA-has-FA-structure} and Lemma \ref{lem:trFA-has-FA-str}.
\end{lemma}
\begin{proof}
 We need to show that the points and the equivalences are preserved, independently of each other.
 The point in $\NR A$ is $\eta(\nil)$, which is mapped to the $\bproj{\unit}$ as required.
 For the equivalence, we only need to check that the underlying functions match accordingly.
 This corresponds to showing commutativity of the following square, for any given $c : A$:
\begin{equation} \label{eq:f-maps}
 \begin{tikzpicture}[x= 1cm , y = -0.4cm,baseline=(current bounding box.center)]
   \node[] (A) at (0,0) {$\NR A$};
   \node[] (B) at (4,0) {$\NR A$};
   \node[] (C) at (0,4) {$\trunc{1}{\FG A}$};
   \node[] (D) at (4,4) {$\trunc{1}{\FG A}$};
   \draw[->] (A) to node[above] {$f_c$ (Lem~\ref{lem:NA-has-equiv})} (B);
   \draw[->] (C) to node[below] {$\bproj{\cons_c}$} (D);
   \draw[->] (A) to node[left] {$k$} (C);
   \draw[->] (B) to node[right] {$k$} (D);
 \end{tikzpicture}
\end{equation}

We do induction on $\alpha : \NR A$.
The goal is an equality in a $1$-type, i.e.\ a set, which means that we only have to check the constructors $\eta$ and $\tau$.
Tracing the explicit construction in Lemma~\ref{lem:FG-has-NR-structure} through the square, we can check directly that the square commutes in both cases (strictly speaking, in the case for $\tau$, it is a cube):
\begin{equation} \label{eq:f-maps1}
 \begin{tikzpicture}[x= 1cm , y = -0.4cm,baseline=(current bounding box.center)]
   \node[] (A) at (0,0) {$\eta(x)$};
   \node[] (B) at (4,0) {$\eta(cx)$};
   \node[] (C) at (0,4) {$\bproj{\vec{\cons_x}}(\bproj \unit)$};
   \node[] (D) at (4,4) {$\bproj{\cons_c}(\bproj{\vec{\cons_x}}(\bproj \unit))$};
   \draw[|->] (A) to node[above] {} (B);
   \draw[|->] (C) to node[below] {} (D);
   \draw[|->] (A) to node[left] {} (C);
   \draw[|->] (B) to node[right] {} (D);
 \end{tikzpicture}
\end{equation}
and:
\begin{equation} \label{eq:f-maps2}
 \begin{tikzpicture}[x= 1cm , y = -0.4cm,baseline=(current bounding box.center)]
   \node[] (A) at (0,0) {$\tau(x,a,y)$};
   \node[] (B) at (4,0) {$\tau(cx,a,y)$};
   \node[] (C) at (0,4) {$\mathsf{ap}_{|\vec{\cons(x)}|}(|\mu_2|(a,y))$};
   \node[] (D) at (4,4) {$\mathsf{ap}_{|\cons_c|(|\vec{\cons(x)}|)}(|\mu_2|(a,y))$};
   \draw[|->] (A) to node[above] {} (B);
   \draw[|->] (C) to node[below] {} (D);
   \draw[|->] (A) to node[left] {} (C);
   \draw[|->] (B) to node[right] {} (D);
 \end{tikzpicture}
\end{equation}
The commutativity is judgmental in the first square, and the second square only uses the usual equality $\mathsf{ap}_g \circ \mathsf{ap}_f = \mathsf{ap}_{g \circ f}$.
\end{proof}

This finally allows us to show:
\mainthm*
\begin{proof}
 By the previous lemma, the composition of the maps in \eqref{eq:two-maps} is an $\FG A$-algebra map.
 But so is the map $\bproj - : \FG A \to \trunc{1}{\FG A}$ by Lemma~\ref{lem:trFA-has-FA-str}.
 Since $\FG A$ is the \emph{initial} such algebra, these two functions must coincide,
 which means that $\bproj - : \FG A \to \trunc 1 {\FG A}$ factors through $\NR A$.
 We know from Lemma~\ref{lem:NA-is-set} that $\NR A$ is a set.
 This implies that $\trunc{1}{\FG A}$ is a set, which is the second part of the theorem.
 
 To see the first part, take $q : \FG A$.
 $\FG A$ having trivial fundamental groups means that $\trunc{0}{q =_{\FG A} q}$ is contractible.
 By~\cite{hott-book}, we have
 \begin{equation}
  \trunc{0}{q =_{\FG A} q} \; \simeq \; \left( \bproj{q} =_{\trunc{1}{\FG A}} \bproj{q}\right).
 \end{equation}
 The second type is contractible since $\trunc{1}{\FG A}$ is a set.
\end{proof}

Having proved the main result, we add two results that now have become very easy:
\begin{lemma} \label{lem:FANA}
 For a set $A$, the two approximations of the free higher group which we have considered are equivalent, i.e.\ $\trunc{1}{\FG A} \simeq \NR A$.
\end{lemma}
\begin{proof}
 From the argument in the proof of the previous theorem, we can follow that $\trunc 1 {\FG A}$ is a retract of $\NR A$.
 Thus, we still need to show that the composition $\NR A \to \trunc 1 {\FG A} \to \NR A$ is the identity.
 But now that we know that everything is a set, it is easy to do this by induction on $\NR A$.
\end{proof}

\begin{theorem}
 The type $\trunc{1}{\FG A}$ is equivalent to the purely set-based free group over $A$ as constructed in \cite[Chp 6.11]{hott-book}.
 If Question~\ref{eq:main-question} can be answered positively, then our free group does indeed generalise the free group construction of \cite{hott-book} from only sets to arbitrary types.
\end{theorem}
\begin{proof}
 Since $\NR A$ is a set by Lemma~\ref{lem:NA-is-set}, it is easy to see that it is equivalent to the set-quotient of $\LS A$ by the relation that identifies a list with the list then one gets after reducing; this is essentially because, when we know that $\NR A$ is a set, the constructors $\swap$ and $\overlap$ become obsolete, and what remains is just this set-quotient.
 But this set-quotient is exactly the purely set-based free group of \cite{hott-book} by \cite[Thm 6.11.7]{hott-book}.
 
 If Question~\ref{eq:main-question} turns out to have a positive answer, then $\FG A$ and $\trunc{1}{\FG A}$ are equivalent, and everything that holds for the latter is also true for the former.
\end{proof}

\section{Conclusions} \label{sec:conclusions}

The central and guiding question of this paper was the problem of showing that the free $\infty$-group $\FG A$ over a set is a set as well.
We have proved a first approximation of this, namely that $\FG A$ has trivial fundamental groups.
This is done entirely in ``book HoTT'', the type theory developed in~\cite{hott-book}.
It would be very interesting to formalise the complete argument in a proof assistant,
and we expect that this would be challenging.
For example, the use of list concatenation in the constructors of the higher inductive type $\NR A$ would lead to many application of \emph{transport} (\emph{substitution}).
It is likely that a different representation of $\NR A$ and the reduction relation would enable a more elegant formalisation.
For a human reader, the presentation in terms of lists is the most intuitive and understandable one that we could think of.

Brunerie has discussed the James construction in homotopy type theory~\cite{DBLP:journals/corr/abs-1710-10307}.
In this context, a type $A$ with a point $\star_A : A$ is given, and the higher inductive type $JA$ is defined to be the \emph{free monoid} over $A$ where $\star_A$ plays the role of the neutral element.
Brunerie then constructs a non-recursive version of $JA$.
Of special interest for him is the case that $A$ is connected (i.e.\ $\trunc{0}{A}$ is contractible), and in this case, $JA$ becomes very similar to our free group.
However, connectedness would be a very unnatural assumption in the present paper; in fact, since we are interested in the case that $A$ is a set, our case of interest is orthogonal to Brunerie's.
If $A$ is not known to be connected, then, compared to our $\FG A$, $JA$ is lacking the condition that every $\cons_a$ is invertible, which is the main source of difficulty in our work.

Related to the current paper is also previous work by Capriotti, Vezzosi, and the current first author~\cite{capKraVez_elimTruncs}.
That work gives a necessary and sufficient condition for a function $X \to Y$ to factor through $\trunc n X$, assuming that $Y$ is $(n+1)$-truncated.
In the current paper, we have been particularly interested in the situation that $n$ is $0$, $X$ is the ``level $0$ approximation'' of the free group (see the description after Definition~\ref{def:NA}), and $Y$ is $\trunc 1 {\FG A}$.
The reason why we have not directly applied the result of \cite{capKraVez_elimTruncs} is that, in our case of interest, showing the mentioned condition is tricky.
This difficulty corresponds to what in our presentation has made the more refined approximation with the constructors $\swap$ and $\overlap$ necessary.
We do not know whether there is an alternative proof of our main result which uses \cite{capKraVez_elimTruncs} directly.

Let us further analyse the methods we have used in the paper.
In principle, the strategy which we have developed should be applicable to more general results than the one we have proved; for example, with some more effort, we expect that it should be possible to show that $\trunc 2 {\FG{A}}$ and $\trunc 3 {\FG{A}}$ (which are better approximations to $\FG{A}$) are sets.
The obvious attempt to do this is to work with a ``better'' non-recursive approximation, i.e.\ a refined version of $\NR A$ which would use higher path constructors to guarantee the coherence of $\swap$ and $\overlap$.
One would then include a $2$- or $3$-truncation constructor instead of the $1$-truncation constructor $\trconst$.
It seems plausible that this could work;
for example, instead of constructing a weakly constant function
\begin{equation} \label{eq:Red-to-eq}
 \Red(x) \to \eta(x) = \eta(\nil)
\end{equation}
as in Lemma~\ref{lem:redIsNeutral-wconst}, we would have to construct a constant function satisfying one or more coherence conditions~\cite{kraus_generaluniversalproperty}, and the new constructors of $\NR A$ would be chosen in such a way that this would be possible.

The additional value that such a generalisation would give us is unclear.
Of course, what we want is to show that $\FG A$ is a set, not just a finite truncation of it.
If we try to use our approach,
it seems we would need to find a way to encode the \emph{whole infinite} tower of coherences in a non-recursive type,
and it looks suspiciously similar to the long-standing open problem of defining semisimplicial types in HoTT~\cite{uf:semisimplicialtypes}.
(To clarify, we would not need a single HIT with infinitely many constructors, since we could take a sequential colimit; and the absence of a general version of Whitehead's principle does not seem to be a problem as long as we can show that \eqref{eq:Red-to-eq} satisfies the coherence conditions given in~\cite{kraus_generaluniversalproperty}, which does not rely on hypercompleteness either.)

Our problem of showing that $\FG A$ is a set is not much different from the open problem of HoTT which asks whether the suspension of a set is a $1$-type; as we have already discussed, what we are asking is essentially whether a the suspension of a set with a distinguished isolated point is a $1$-type.
Thus, our question is slightly weaker and, as far as we can see, an answer to the weaker question would not be sufficient to answer the more general question.\footnote{Related is the discussion \emph{Does ``adding a path'' preserve truncation levels?} at \url{https://groups.google.com/forum/\#!topic/homotopytypetheory/gVmcvaOeD5c}.}
However, it seems plausible that an approach similar to ours is applicable to the more general question as well.
In this case, being able to encode infinite towers of coherences could potentially be key to both open problems, although of course there would still be a lot of work to do (which might or might not even be impossible).

Theories such as Voevodsky's \emph{homotopy type system} (HTS)~\cite{voe:hts}, \emph{two-level type theory} (2LTT)~\cite{alt-cap-kra:two-level,ann-cap-kra:two-level} or \emph{computational higher type theory}~\cite{2017arXiv171201800A} are variations of standard HoTT in which such infinite structures can be constructed.
We believe it would be worth investigating whether the ``suspension of a set'' problem can be resolved in such systems.
Our preliminary investigations hint that it is at least be possible to define a ``completely non-recursive'' version of $\FG A$, which would be a starting point.
However, actually \emph{using} this construction to mimic the proof that we have given in this paper is, of course, a completely different story.

If (we are now in the realm of complete speculation) it turns out that HTS can prove that the suspension of a set is a $1$-type, it would be even more interesting whether ``standard HoTT'' can do it as well.
This is because one would need to come up with a completely different argument in standard HoTT, and if it turns out that the open problem is independent of standard HoTT, hope for a conservativity result would be lost for all theories that are powerful enough to encode semisimplicial types.
Recall that we have a conservativity result for 2LTT, due to Capriotti~\cite{paolo:thesis}, which says that a fibrant type in 2LTT can only be inhabited if the corresponding type in HoTT (assuming it exists) is inhabited as well.
This however only works for a version of 2LTT where we do \emph{not} have semisimplicial types in the usual formulation (we only have semisimplicial types indexed over the pretype of strict natural numbers, but not over the type fibrant natural numbers).
Thus, in this version of 2LTT, the sketched approach to solve the open problem regarding the suspension of a set would not work.
This might be more than a coincidence.

\subsubsection*{Acknowledgements}
 We are very grateful to Paolo Capriotti for many discussions on the topic, and for several remarks which have influenced this paper.
 Most importantly, the construction of a free group using the constructors of the free monoid and conditions ensuring that $\cons_a$ is an equivalence is due to Paolo, as well as the decomposition shown in~\eqref{eq:adj-diag}.
 It was also him who suggested using the statement of Lemma~\ref{lem:map-is-double-alg} to complete the main proof of the paper,
 which has led to a cleaner proof than if we had done it with a more direct argument.
 We further thank Rafa\"el Bocquet for discussions on free groups,
 and the anonymous reviewers whose comments have helped us to improve the presentation and readability of the paper.

\bibliographystyle{plain}
\bibliography{master}

\end{document}